\documentclass[article,10pt,onecolumn]{IEEEtran}
\usepackage{epsfig}
\usepackage{amsmath}
\usepackage{amstext}
\usepackage{amsfonts}
\usepackage{amssymb}
\usepackage{eucal}
\usepackage{graphicx}
\usepackage{verbatim}
\usepackage{stmaryrd}
\usepackage{pstricks}
\usepackage{pst-node}
\usepackage{psfrag}
\usepackage{arydshln}
\usepackage{lineno}

\DeclareMathOperator{\diag}{diag}

\DeclareMathOperator{\sgn}{sgn}

\newcommand{\C}{C}

\newcommand{\tr}{{\rm tr}}
\newcommand{\defeq}{\overset{\Delta}{=}}

\newtheorem{lemma}{Lemma}

\newtheorem{theorem}{Theorem}

\begin{document}
\bibliographystyle{IEEEtran}

\title{A Bayesian Framework for Collaborative Multi-Source Signal Detection}

%\author{\authorblockN{Romain Couillet}
%\authorblockA{ST-Ericsson, Sup\'elec\\
%635 Route des Lucioles\\
%06560 Sophia Antipolis, France\\
%Email: romain.couillet@supelec.fr}
%\and
%\authorblockN{M{\'e}rouane~Debbah}
%\authorblockA{Alcatel-Lucent Chair, Sup\'elec\\
%Plateau de Moulon, 3 rue Joliot-Curie \\
%91192 Gif sur Yvette, France\\
%Email: merouane.debbah@supelec.fr}
%}
\author{Romain~Couillet$^\dagger$,~\IEEEmembership{Student Member,~IEEE}\thanks{$^\dagger$ST-Ericsson - Sophia Antipolis, 635 Route des Lucioles, 06560 Sophia Antipolis, romain.couillet@supelec.fr}, M{\'e}rouane~Debbah$^\ast$,~\IEEEmembership{Senior Member,~IEEE}\thanks{$^\ast$Chair Alcatel Lucent, Sup\'elec - 3 rue Joliot Curie, 91192 Gif sur Yvette, merouane.debbah@supelec.fr. Debbah's work is partially supported by the European Commission in the framework of the FP7 Network of Excellence in Wireless Communications NEWCOM++.}
}

%\markboth{}{}
\maketitle

\begin{abstract}
This paper introduces a Bayesian framework to detect multiple signals embedded in noisy observations from a sensor array. For various states of knowledge on the communication channel and the noise at the receiving sensors, a marginalization procedure based on recent tools of finite random matrix theory, in conjunction with the maximum entropy principle, is used to compute the hypothesis selection criterion. Quite remarkably, explicit expressions for the Bayesian detector are derived which enable to decide on the presence of signal sources in a noisy wireless environment. The proposed Bayesian detector is shown to outperform the classical power detector when the noise power is known and provides very good performance for limited knowledge on the noise power. Simulations corroborate the theoretical results and quantify the gain achieved using the proposed Bayesian framework.
\end{abstract}

\section{Introduction}
Since a few years, the idea of smart receiving devices has made its way through the general framework of cognitive radio \cite{MIT99}. The general idea of an ideal cognitive receiver is a device that is capable of inferring on any information it is given to discover by itself the surrounding environment \cite{HAY05}. Such a device should be first able to turn prior information on the transmission channel into a mathematically tractable form. This allows then the terminal to take optimal instantaneous decisions in terms of information to feed back, bandwidth to occupy, transmission power to use etc. It should also be capable of updating its knowledge to continuously adapt to the dynamics of the environment. This vision of a cognitive radio is compliant with Haykin's definition of ``brain empowered'' wireless devices \cite{HAY05}.

In particular, one of the key features of cognitive receivers is their ability to {\it sense} free spectrum. Indeed, when the cognitive device is switched on, its prior knowledge is very limited but still it is requited to {\it decide} whether it receives informative data
or pure noise due to interfering background electromagnetic fields, on different frequency bands: this will be further referred to as the \textit{signal detection} procedure.

In the single-input single-output antenna (SISO) scenario, the study of Bayesian signal detectors dates back to the work of Urkowitz \cite{URK67} on additive white Gaussian noise (AWGN) channels. It was later extended to more realistic channel models \cite{KOS02}-\cite{DIG03}. Urkowitz's signal detector is optimal in the sense that his process performs the maximum \textit{correct detection rate}, i.e. the odds for an informative signal to be detected as such, for a given low \textit{false alarm rate}, i.e. the odds for a pure noise input to be wrongly declared an informative signal. To the authors' knowledge, the multi-antenna (MIMO) extension has not been studied, because of the almost prohibitive mathematical complexity of the problem. The usual power detection technique from Urkowitz was then adapted to the MIMO scenario, e.g. \cite{POO08}. The latter consists in summing up all the individual powers received at each antenna and declaring that the incoming signal carries {\it information}\footnote{in order not to confuse the notions of received signal and transmitted signal other than white noise, the hypothetically transmitted signal will be referred to as the {\it information} or {\it informative signal}.} if the total received power exceeds a given threshold; otherwise the received signal is declared pure noise. Therefore this technique does not capitalize on the knowledge of whatever prior information the receiver might be aware of, apart from an approximate estimation of the signal to noise ratio (SNR) required to preset the decision threshold. %Alternatively, subspace or eigenvalue-based methods are used, such as in \cite{pascal}, where a generalized likelihood ratio threshold technique is used to cope with the problem of unknown noise variance and channel; quite surprisingly, in this work which assumes large numbers of sensors and sampling instants, the decision criterion reduces to estimating the ratio between the largest eigenvalue and the trace of the received sample covariance matrix. Some contributions treat the problem of estimating the number of transmitting sources embedded in noise, e.g. \cite{wax}, \cite{hero}, which is more specific than our more general ``signal against noise'' decision criterion. 

This raises the interest for new techniques such as cooperative spectrum sensing using multiple antennas \cite{POO08}. Those techniques propose to improve the signal detection method of Urkowitz by using extra system dimensions (space dimension through cooperation among terminals for example). Unfortunately, the approaches used are highly dependent on the initial assumptions made and have led to many different contributions. For instance, some insightful work emerged which uses eigen-spectrum or subspace analysis of the received sampled signals \cite{CAR08}-\cite{ZEN08}. Those might provide interesting results in their simplicity and their limited need for prior system knowledge. However, those studies usually consider static knowledge at the receiver and do not cope with the fact that very limited information is indeed provided to the sensing device; this static knowledge often includes the {\it a priori} exact or approximative knowledge of the SNR.\footnote{note that even the ratio $\lambda_{\rm max}/\lambda_{\rm min}$ of the extreme eigenvalues of the Gram matrix of the received signal matrix, which is independent of the SNR when the latter is asymptotically large, is in practical finite-dimensional applications strongly dependent on the SNR.} As an answer to the challenging problem of {\it a priori} limited channel and noise information, \cite{pascal} proposes a sub-optimal technique based on the generalized likelihood ratio threshold (GLRT); this threshold is shown to merely consist of the ratio between the largest eigenvalue and the trace of the receive empirical covariance matrix and provides a significant improvement compared to the classical detection method based on the conditioning number of that matrix. It is also worth mentioning that some contributions treat the problem of estimating the number of transmitting sources embedded in noise, e.g. \cite{wax}, \cite{hero}, which is more specific than our more general ``signal against noise'' decision criterion; this problem can in fact be treated using the results of this paper, with no need for a complete problem redefinition (see Section \ref{sec:discussion}). 

In this work, we introduce a general Bayesian framework for signal detection, which is consistent with the receiver's state of knowledge on the environment, i.e. which produces a unique transmission model for each prior state of knowledge; this knowledge being a list of statistical constraints on the transmission environment. The methodology relies more precisely on the work of Jaynes \cite{JAY03} on Bayesian probabilities and especially on the maximum entropy principle. This principle shall allow us to provide a {\it consistent} detection criterion for each set of prior knowledge at the receiver; here the term {\it consistent} must be understood in the sense that any alternative method would not be fully compliant (or maximally honest in the own terms of Jaynes) with the prior information at the receiver. 

This paper is structured as follows: In Section \ref{sec:results} we introduce the scope and the main results of this paper. In Section \ref{sec:model} we formulate the signal detection model. Then in Section \ref{sec:sigdet}, the Bayesian signal detectors are computed for different levels of knowledge on the system model.  Simulations are then presented in Section \ref{sec:simu}. Finally, after a short discussion in Section \ref{sec:discussion} on the general framework and its limitations, we provide our conclusions.

{\it Notations}: In the following, boldface lowercase and uppercase characters are used for vectors and matrices, respectively. We note $(\cdot)^{\sf H}$ the Hermitian transpose, $\tr(\cdot)$ denotes the matrix trace. $\mathcal M(\mathcal A,N,M)$ is the set of matrices of size $N\times M$ over the algebra $\mathcal A$. $\mathcal U(N)$ is the set of unitary square matrices of size $N$. The notation $P_X(Y)$ denotes the probability density function of the variable $X$ evaluated in the vicinity of $Y$. The notation $(x)_+$ equals $x$ if $x>0$ and $0$ otherwise.

\section{Main Results}
\label{sec:results}
The main purpose of this work is to propose a universal framework to signal detection, based on cogent information at the multi-antenna sensing device. That is, the prior information at the receiver, before channel sensing, will be summarized as a set of statistical information about the transmission environment. We will consider here the situation when the additive noise power is known perfectly or known to belong to a bounded interval, and the situation when the number of signal sources is either known or known to be less than a maximum.

This information will be translated into its corresponding {\it most appropriate} mathematical model, in the sense that this model is (i) compliant with the prior information at the receiver and (ii) avoid enforcing empirical (therefore unknown) properties. This can be realized thanks to the maximum entropy principle.

The decision criterion will then be based on the ratio

\begin{equation}
  C_{{\bf Y}|I}({\bf Y}) = \frac{P_{\mathcal H_1|{\bf Y},I}}{P_{\mathcal H_0|{\bf Y},I}}
\end{equation}
between the hypothesis $\mathcal H_1$ of the received data $\bf Y$ containing an information signal and the hypothesis $\mathcal H_0$ of its containing only noise, when the prior information $I$ (and then its associated maximum entropic model) is considered. Evaluating $C_{ {\bf Y}|I}$ when the system parameters are multi-dimensional could be done by numerical approximation of the underlying integral formulas, but then the `curse of dimensionality' quickly arises and results become very inaccurate, already for little dimensions. This obliges one to explicitly compute integrals of (possibly large) matrix parameters, which is performed here thanks to latest advances in the field of finite-dimension random matrix theory.

The main results of this work are summarized in two theorems and some simulation-based observations. First we consider the situation when the multi-antenna sensing device is entitled to decide on the presence of a single transmitting signal source on a narrow-band channel, assumed static for a given (possibly short) sensing period. This situation is referred to as the SIMO scenario. In theorem \ref{th:1}, a closed-form expression of $C_{{\bf Y}|I}$ is derived when the maximum entropic model attached to $I$ imposes independent Gaussian channel, signal and noise entries; this expression turns out to be solely dependent on the eigenvalues of the matrix ${\bf YY}^{\sf H}$, but in a more involved form than the classical power detector. Monte Carlo simulations, supposing an accurate evaluation of the transmission model, will in fact show a large detection gain of the novel Bayesian detector compared to the classical energy detector. For systems which tolerate very low false alarm decisions, e.g. in a cognitive radio setup when secondary users are banned to transmit in bands in use by primary users, the detection gain of our framework are in particular observed to be as large as $10\%$ (and possibly more if more numerous simulations are run).

Theorem \ref{th:2} generalizes theorem \ref{th:1} to the scenario when multiple sources are transmitting and their exact number is known to the sensing receiver. This is therefore referred to as the MIMO scenario. Simulations reveal less accurate decision capabilities in this scenario compared to SIMO. This is interpreted as a consequence of the increased number of variables whose joint distribution is less constrained in the MIMO case than in the SIMO case; this makes a specific realization of $\bf Y$ more difficult to fit to the underlying MIMO channel model. The MIMO decision is still more efficient than that of the energy detector but now the gap between both closes in.

The most interesting feature of our Bayesian signal detection framework appears when we treat the problem of imprecise knowledge about the SNR. In most scientific contributions about signal detection techniques, one assumes perfect or close-to-perfect knowledge about the noise power, which is obviously inaccessible to the sensing device, whose objective is precisely to separate pure noise from informative signals. We derive in this work a formula for evaluating the presence of an informative signal when the knowledge about the noise power may be very limited. It turns out that, even when the {\it a priori} noise distribution spans uniformly from an (unrealistically) low value to an (unrealistically) large value, simulations results suggest that the proposed signal detector performs remarkably well. This is in sharp contrast with classical methods, such as the energy detector, subspace-based methods or extreme eigenvalue based methods, for which the information about the noise level is required to some extent. This also allows us to go further than in \cite{pascal}, where the GLRT technique was used to cope with unknown channel (modelled as a matrix $\bf H$) and noise information (gathered into the noise power $\sigma^2$); the decision ratio targeted for the GLRT consists in the ratio between $\sup_{{\bf H},\sigma^2}P_{\mathcal H_1|{\bf Y},I}$ and $\sup_{{\bf H},\sigma^2}P_{\mathcal H_1|{\bf Y},I}$; this is in particular inaccurate when the {\it a priori} distribution for the random matrix $\bf H$ is `broad' around the effective value of the channel matrix. Finally, trivial applications of Bayes' rule enable the extension of the current signal detection problem to a wider range of detection issues, such as the evaluation of the number of transmitting sources.

\section{Signal Model}
\label{sec:model}
We consider a simple communication system composed of $M$ transmitter sources, e.g. this can either be an $M$-antenna single transmitter or $M$ single antenna (not necessarily uncorrelated) transmitters, and a receiver composed of $N$ sensors, be they the (uncorrelated) antennas of a single terminal or a mesh of scattered sensors. To enhance the multiple-antenna (MIMO) analogy model, the joint set of sources and the joint set of sensors will often be referred to as {\it the transmitter} and {\it the receiver}, respectively. The transmission channel between the transmitter and the receiver is modelled by the matrix ${\bf H}\in \mathbb C^{N\times M}$, with entries $h_{ij}$, $1\leq i\leq N$, $1\leq j\leq M$. If, at time $l$, the transmitter emits data, those are denoted by the $M$-dimensional vector ${\bf s}^{(l)}=(s_1^{(l)},\ldots,s_M^{(l)})^{\sf T}\in \mathbb C^M$. The additive white Gaussian noise at the receiver is modelled, at time $l$, as the $N$-dimensional vector $\sigma{\bf \theta}^{(l)}=\sigma(\theta_1^{(l)},\ldots,\theta_N^{(l)})^{\sf T}\in \mathbb C^N$, with $\sigma^2$ the variance of the noise vector entries. For simplicity in the following derivations, we shall consider unit variance of the entries of both ${\bf \theta}^{(l)}$ and ${\bf s}^{(l)}$, i.e. ${\rm E}[|\theta^{(l)}_i|^2]=1$, ${\rm E}[|s^{(l)}_i|^2]=1$. We then denote ${\bf y}^{(l)}=(y_1^{(l)},\ldots,y_N^{(l)})^{\sf T}$ the $N$-dimensional data received at time $l$. Assuming the channel coherence time is at least as long as $L$ sampling periods, we finally denote ${\bf Y}=({\bf y}^{(1)},\ldots,{\bf y}^{(L)})\in \mathbb C^{N\times L}$ the matrix of the concatenated receive vectors. This scenario is depicted in Figure \ref{fig:sensors}.

\begin{figure}[t]
  \centering
\includegraphics[]{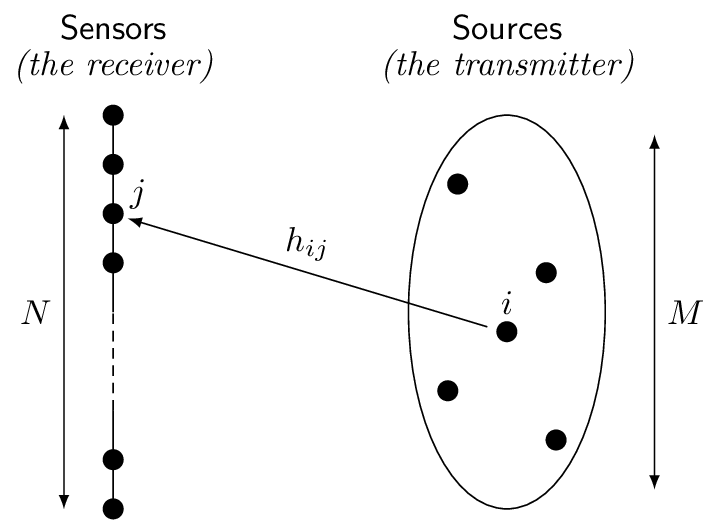}
  \caption{MIMO signal detection configuration}
  \label{fig:sensors}
\end{figure}

Depending on whether the transmitter emits signals, we consider the following hypotheses
\begin{itemize}
 \item $\mathcal H_0$. Only background noise is received.
 \item $\mathcal H_1$. Informative signals plus background noise are received.
\end{itemize}

Therefore, under condition $\mathcal H_0$, we have the model,

\begin{equation}
\label{eq:H0}
{\bf Y} 
= \sigma \begin{pmatrix} \theta^{(1)}_1 & \cdots & \theta^{(L)}_1 \\ \vdots & \ddots & \vdots \\ \theta^{(l)}_N & \cdots & \theta^{(L)}_N \end{pmatrix}
\end{equation}
and under condition $\mathcal H_1$,

\begin{equation}
\label{eq:H1}
{\bf Y} =
\begin{pmatrix}
h_{11} & \ldots & h_{1M} & \sigma & \cdots & 0 \\
\vdots & \vdots & \vdots & \vdots & \ddots & \vdots \\
h_{N1} & \ldots & h_{NM} & 0 & \cdots & \sigma
\end{pmatrix}
\begin{pmatrix}
s^{(1)}_1 & \cdots & \cdots &  s^{(L)}_1 \\
\vdots & \vdots & \vdots & \vdots \\
s^{(1)}_M & \cdots & \cdots &  s^{(L)}_M \\
\theta^{(1)}_1 & \cdots & \cdots & \theta^{(L)}_1 \\
\vdots & \vdots & \vdots & \vdots \\
\theta^{(1)}_N & \cdots & \cdots & \theta^{(L)}_N
\end{pmatrix}
\end{equation}

Under this hypothesis, we further denote $\bf \Sigma$ the covariance matrix of ${\bf YY}^{\sf H}$,

\begin{align}
\label{eq:Sigma}
{\bf \Sigma} &= {\rm E}[{\bf Y}{\bf Y}^{\sf H}] \\ &= L \left({\bf H}{\bf H}^{\sf H}+\sigma^2{\bf I}_{N} \right)\\
\label{eq:ULU} &= {\bf U}\left(L \bf{\Lambda}\right){\bf U}^{\sf H}
\end{align}
where $\bf \Lambda = \diag\left(\nu_1+\sigma^2,\ldots,\nu_N+\sigma^2 \right)$, with $\left\{\nu_i,i\in \{1,\ldots,N\}\right\}$ the eigenvalues of ${\bf HH}^{\sf H}$ and $\bf U$ a certain unitary matrix. 

The receiver is entitled to decide whether the base station is transmitting informative signals or not; this is, the receiver makes a decision over hypothesis $\mathcal H_0$ or $\mathcal H_1$. The receiver is however considered to have very limited information about the transmission channel and is in particular not necessarily aware of the exact number $M$ of sources and of the signal-to-noise ratio. For this reason, the maximum entropy principle requires that all unknown variables be assigned a probability distribution which is both (i) consistent with the prior information (voluntarily discarding information violates the Bayesian philosophy) and (ii) has maximal entropy over the set of densities that validate (i). It is known in particular that, upon the unique constraint of its covariance matrix, the entropy maximizing joint distribution of a given vector is central Gaussian. Therefore, if ${\bf H}$ is only known to satisfy, as is often the case in the short term, ${\rm E}[{\rm tr}{\bf HH}^{\sf H}]=1$, the maximum entropy principle states that the $h_{ij}$'s are independent and distributed as $h_{ij}\sim \mathcal{CN}(0,1/M)$. For the same reason, both noise $\theta^{(l)}_{i}$ and signal $s^{(l)}_i$ variables are taken as independent central Gaussian with variance ${\rm E}[|\theta^{(l)}_{i}|^2]=1$, ${\rm E}[|s^{(l)}_{i}|^2]=1$.

The decision criterion for the receiver to establish whether an informative signal was transmitted is based on the ratio $C$,

\begin{equation}
C({\bf Y})=\frac{P_{\mathcal H_1|{\bf Y}}({\bf Y})}{P_{\mathcal H_0|{\bf Y}}({\bf Y})}
\end{equation}

Thanks to Bayes' rule, this is

\begin{equation}
C({\bf Y})=\frac{P_{\mathcal H_1}\cdot P_{{\bf Y}|\mathcal H_1}({\bf Y})}{P_{\mathcal H_0}\cdot P_{{\bf Y}|\mathcal H_0}({\bf Y})}
\end{equation}
with $P_{\mathcal H_i}$ the {\it a priori} probability for hypothesis $\mathcal H_i$ to be true. We suppose that no side information allows the receiver to think $\mathcal H_1$ is more or less probable than $\mathcal H_0$, and therefore set $P_{\mathcal H_1}=P_{\mathcal H_0}=\frac{1}{2}$ (this again validates the maximum entropy principle), and then

\begin{equation}
C({\bf Y})=\frac{P_{{\bf Y}|\mathcal H_1}({\bf Y})}{P_{{\bf Y}|\mathcal H_0}({\bf Y})}
\end{equation}
reduces to a maximum likelihood criterion.

In the next section, we will derive close-form expressions for $C({\bf Y})$ under the hypotheses that the values of $M$ and the SNR are either perfectly or only partially known at the receiver.

\section{Signal Detection}
\label{sec:sigdet}

\subsection{Known noise variance and number of signal sources}
\subsubsection{Derivation of $P_{{\bf Y}|\mathcal H_i}$ in SIMO case}
We first analyze the situation when the noise power $\sigma^2$ and the number $M$ of signal sources are known to the receiver. We also assume in this first scenario that $M=1$. Further consider that $L>N$, which is a common assumption.

\paragraph{Pure noise likelihood $P_{{\bf Y}|\mathcal H_0}$}
In this first scenario, the noise entries $\theta^{(l)}_i$ are Gaussian and independent. The distribution for $\bf Y$, that can be seen as a random vector with $NL$ entries, is then a $NL$ multivariate uncorrelated complex Gaussian with covariance matrix $\sigma^2 {\bf I}_{NL}$,

\begin{align}
P_{{\bf Y}|{\mathcal H_0}}({\bf Y}) &= \frac{1}{(\pi\sigma^2)^{NL}}e^{-\frac{1}{\sigma^2}\tr{}{{\bf YY}^{\sf H}}} \\
\end{align}
by denoting ${\bf x}=(x_1,\ldots,x_N)^{\sf T}$ the eigenvalues of ${\bf YY}^{\sf H}$, \eqref{eq:PH0} only depends on $\sum_{i=1}^N x_i$,

\begin{align}
\label{eq:PH0}
P_{{\bf Y}|{\mathcal H_0}}({\bf Y}) &= \frac{1}{(\pi\sigma^2)^{NL}}e^{-\frac{1}{\sigma^2}\sum_{i=1}^N x_i}
\end{align}

\paragraph{Informative signal likelihood $P_{{\bf Y}|\mathcal H_1}$}
In scenario $\mathcal H_1$, the problem is more involved. The entries of the channel matrix $\bf H$ were modelled as jointly uncorrelated Gaussian distributed, with ${\rm E}[|h_{ij}|^2]=1/M$. Therefore, since $M=1$, ${\bf H}\in \mathbb C^{N\times 1}$ and ${\bf \Sigma}={\bf HH}^{\sf H}+\sigma^2{\bf I}_N$ has $N-1$ eigenvalues equal to $\sigma^2$ and another distinct eigenvalue $\lambda_1=\nu_1+\sigma^2=(\sum_{i=1}^N|h_{i1}|^2)+\sigma^2$. The density of $\lambda_1-\sigma^2$ is a complex $\chi^2_N$ distribution (which is, up to a scaling factor $2$, equivalent to a real $\chi^2_{2N}$ distribution). Hence the eigenvalue distribution of $\bf \Sigma$, defined on $\mathbb R^{+N}$,

\begin{align}
P_{{\bf \Lambda}}({\bf \Lambda})= 
& \label{eq:Lambda} \frac{1}{N}(\lambda_1-\sigma^2)_+^{N-1}\frac{e^{-(\lambda_1-\sigma^2)}}{(N-1)!}\prod_{i=2}^N \delta(\lambda_i-\sigma^2)
\end{align}

Given model \eqref{eq:H1}, $\bf Y$ is distributed as correlated Gaussian, 

\begin{align}
  \label{eq:YgivenSigma}
  P_{{\bf Y}|{\bf \Sigma},I_1}({\bf Y},{\bf \Sigma}) &= \frac{1}{\pi^{LN}\det({\bf \Lambda})^L}e^{-\tr\left({\bf YY}^{\sf H}{\bf U\Lambda}^{-1}{\bf U}^{\sf H}\right)}
\end{align}
where $I_k$ denotes the prior information ``$\mathcal H_1$ and $M=k$''.

Since the channel $\bf H$ is unknown, we need to integrate out all possible channels of the model \eqref{eq:H1} over the probability space of $N\times M$ matrices with Gaussian i.i.d. distribution. This is equivalent to integrating out all possible covariance matrices $\bf \Sigma$ over the space of such positive definite Hermitian matrices

\begin{align}
  \label{eq:intSigma}
  P_{{\bf Y}|\mathcal H_1}({\bf Y}) &= \int_{{\bf \Sigma}} P_{{\bf Y}|{\bf \Sigma}\mathcal H_1}({\bf Y},{\bf \Sigma})P_{\bf \Sigma}({\bf \Sigma})d{\bf \Sigma}
\end{align}

In the following, we shall prove that the integral \eqref{eq:intSigma} can be derived as the following closed-form expression,
\begin{theorem}
  \label{th:1}
  The detection ratio $C_{ {\bf Y}|I_1}({\bf Y})$ for the presence of an informative signal under prior information $I_1$, i.e. the receiver knows (i) $M=1$ signal source, (ii) the SNR $\sigma^{-2}$, reads

\begin{equation}
\label{eq:C_SIMO}
C_{{\bf Y}|I_1}({\bf Y}) = \frac1N\sum_{l=1}^N \frac{\sigma^{2(N+L-1)}e^{\sigma^2+\frac{x_l}{\sigma^2}}}{{\prod_{\substack{i=1 \\ i\neq l}}^N(x_l-x_i)}} J_{N-L-1}(\sigma^2,x_l) 
\end{equation}
with $x_1,\ldots,x_N$ the empirical eigenvalues of ${\bf YY}^{\sf H}$ and where 

\begin{align}
J_k(x,y) &= \int_{x}^{+\infty}t^{k}e^{-t-\frac{y}{t}}dt
\end{align}
\end{theorem}
\begin{proof}
  We start by noticing that $\bf H$ is Gaussian and therefore the joint density of its entries is invariant by left and right unitary products. As a consequence, the distribution of the matrix ${\bf \Sigma}={\bf HH}^{\sf H}+\sigma^2{\bf I}$ is unitarily invariant, i.e. for any unitary matrix ${\bf V}$, ${\bf V\Sigma V}^{\sf H}$ has the same joint density as $\bf \Sigma$. The latter density does not as a consequence depend on $\bf U$ in its singular value decomposition \eqref{eq:ULU}. This allows us to write, similarly as in \cite{GUI06},

\begin{align}
P_{{\bf Y}|\mathcal H_1}({\bf Y}) 
&= \int_{{\bf \Sigma}} P_{{\bf Y}|{\bf \Sigma},\mathcal H_1}({\bf Y},{\bf \Sigma})P_{\bf \Sigma}({\bf \Sigma})d{\bf \Sigma} \\
&= \int_{\mathcal U(N)\times \mathbb {R^+}^N} P_{{\bf Y}|{\bf \Sigma},\mathcal H_1}({\bf Y},{\bf \Sigma})P_{L{\bf \Lambda}}(L{\bf \Lambda})d{\bf U}d(L{\bf \Lambda}) \\
&= \label{eq:bigint} \int_{\mathcal U(N)\times \mathbb {R^+}} P_{{\bf Y}|{\bf \Sigma},\mathcal H_1}({\bf Y},{\bf \Sigma})P_{{\lambda_1}}({\lambda_1})d{\bf U}d{\lambda_1}
\end{align}

Equation \eqref{eq:bigint} leads then to

\begin{align}
	P_{{\bf Y}|I_1}({\bf Y}) &= \int_{\mathcal U(N)\times \mathbb {R^+}^N} \frac{1}{\pi^{NL}\det({\bf \Lambda})^L}e^{-\tr{}{\left({\bf YY}^{\sf H}{\bf U\Lambda}^{-1}{\bf U}^{\sf H}\right)}} (\lambda_1-\sigma^2)_+^{N-1}\frac{e^{-(\lambda_1-\sigma^2)}}{N!}\prod_{i=2}^N \delta(\lambda_i-\sigma^2) d{\bf U}d{\lambda_1}\ldots d{\lambda_N}
\end{align}

To go further, we utilize the Harish-Chandra identity \cite{BAL00}

\begin{equation}
\int_{\mathcal U(N)}e^{\kappa \tr{}{({\bf AUBU}^{\sf H}})}d{\bf U}=\left(\prod_{n=1}^{N-1}n!\right)\kappa^{N(N-1)/2}\frac{\det\left(\left\{e^{-A_iBj}\right\}_{\substack{1\leq i\leq N\\1\leq j\leq N}}\right)}{\Delta({\bf A})\Delta({\bf B})}
\end{equation}
in which, for a matrix $\bf X$ with eigenvalues $x_1,\ldots,x_N$, $\Delta({\bf X})$ indicates the Vandermonde determinant

\begin{equation}
\Delta({\bf X})=\prod_{i>j}(x_i-x_j)
\end{equation}

And then $P_{{\bf Y}|I_1}({\bf Y})$ further develops as

\begin{align}
P_{{\bf Y}|I_1}({\bf Y}) &= \lim_{\lambda_2,\ldots,\lambda_N\rightarrow \sigma^2} \frac{e^{\sigma^2}(-1)^{\frac{N(N-1)}{2}}\prod_{j=1}^{N-1}j!}{\pi^{LN}\sigma^{2L(N-1)}N!}\int_{\sigma^2}^{+\infty}\frac{1}{\lambda_1^L}(\lambda_1-\sigma^2)^{N-1}e^{-\lambda_1}\frac{\det\left(\left\{e^{-\frac{x_i}{\lambda_j}}\right\}_{\substack{1\leq i\leq N \\ 1\leq j\leq N}}\right)}{\Delta({\bf X})\Delta({{\bf \Lambda}^{-1}})}d\lambda_1 \\
\label{eq:Deltam1} &= \lim_{\lambda_2,\ldots,\lambda_N\rightarrow \sigma^2} \frac{e^{\sigma^2}\prod_{j=1}^{N-1}j!}{\pi^{LN}\sigma^{2L(N-1)}N!}\int_{\sigma^2}^{+\infty}\frac{1}{\lambda_1^L}(\lambda_1-\sigma^2)^{N-1}e^{-\lambda_1}\det\left({\bf \Lambda}^{N-1}\right)\frac{\det\left(\left\{e^{-\frac{x_i}{\lambda_j}}\right\}_{\substack{1\leq i\leq N \\ 1\leq j\leq N}}\right)}{\Delta({\bf X})\Delta({{\bf \Lambda}})} d\lambda_1 \\
\label{eq:det} &= \lim_{\lambda_2,\ldots,\lambda_N\rightarrow \sigma^2} \frac{e^{\sigma^2}\sigma^{2(N-1)(N-L-1)}\prod_{j=1}^{N-1}j!}{\pi^{LN}N!}\int_{\sigma^2}^{+\infty}{\lambda_1}^{N-L-1}(\lambda_1-\sigma^2)^{N-1}e^{-\lambda_1}\frac{\det\left(\left\{e^{-\frac{x_i}{\lambda_j}}\right\}_{\substack{1\leq i\leq N \\ 1\leq j\leq N}}\right)}{\Delta({\bf X})\Delta({\bf \Lambda})} d\lambda_1
\end{align}
in which $\bf X$ and $x_1,\ldots,x_N$ respectively correspond to ${\bf YY}^{\sf H}$ and its eigenvalues.
The equality \eqref{eq:Deltam1} comes from the fact that $\Delta({\bf \Lambda}^{-1})=(-1)^{N(N+3)/2}\frac{\Delta({\bf \Lambda})}{\det({\bf \Lambda})^{N-1}}$.

By denoting ${\bf y}=(y_1,\ldots,y_{N-1},y_N)=(\lambda_2,\ldots,\lambda_N,\lambda_1)$ and the functions,

\begin{align}
f(x_i,y_j) &= e^{-\frac{x_i}{y_j}} \\
f_i(y_j) &= f(x_i,y_j)
\end{align}
we can perform a similar derivation as in \cite{SIM06} to obtain

\begin{align}
\lim_{\lambda_2,\ldots,\lambda_N\rightarrow \sigma^2} \frac{\det\left(\left\{e^{-\frac{x_i}{\lambda_j}}\right\}_{\substack{1\leq i\leq N \\ 1\leq j\leq N}}\right)}{\Delta({\bf X})\Delta({\bf \Lambda})} &= \lim_{\substack{y_1,\ldots,y_{N-1}\rightarrow \sigma^2 \\ y_N\rightarrow \lambda_1}} (-1)^{N-1} \frac{\det\left( \left\{ f_i(x_j) \right\}_{i,j}\right) }{\Delta({\bf X})\Delta({\bf \Lambda})} \\
&= (-1)^{N-1} \frac{\det\left[ f_i(\sigma^2),\ f_i'(\sigma^2),\ldots,\ f^{(N-2)}(\sigma^2),\ f_i(\lambda_1) \right]}{\prod_{i<j}(x_i-x_j)(\lambda_1-\sigma^2)^{N-1}\prod_{j=1}^{N-2}j!}
\end{align}

The change of variables led to a switch of one column and explains the $(-1)^{N-1}$ factor when computing the resulting determinant.

The partial derivatives of $f$ along the second variable is

\begin{align}
\left(\frac{\partial}{\partial y^{k}} f\right)_{k\geq 1}(a,b) &= \sum_{m=1}^k \frac{(-1)^{k+m}}{b^{m+k}}\C_k^m\frac{(k-1)!}{(m-1)!}a^me^{-\frac{a}{b}} \\
&\defeq \kappa_k(a,b)e^{-\frac{a}{b}}
\end{align}

Back to the full expression of $P_{{\bf Y}|\mathcal H_1}({\bf Y})$, we then have

\begin{align}
P_{{\bf Y}|I_1}({\bf Y}) &= \frac{e^{\sigma^2}\sigma^{2(N-1)(N-L-1)}}{N\pi^{LN}} \int_{\sigma^2}^{+\infty}(-1)^{N-1}\lambda_1^{N-L-1}e^{-\lambda_1} \frac{\det\left[ f_i(\sigma^2),\ f_i'(\sigma^2),\ldots,\ f^{(N-2)}(\sigma^2),\ f_i(\lambda_1) \right]}{\prod_{i<j}(x_i-x_j)}d\lambda_1 \\
&= \frac{e^{\sigma^2}\sigma^{2(N-1)(N-L-1)}}{N\pi^{LN}\prod_{i<j}(x_i-x_j)}\int_{\sigma^2}^{+\infty} (-1)^{N-1}\lambda_1^{N-L-1}e^{-\lambda_1} \det\left[
\begin{array}{c:c:c}
	\begin{matrix} e^{-\frac{x_1}{\sigma^2}} \\ \vdots \\ e^{-\frac{x_N}{\sigma^2}} \end{matrix} & \left(\kappa_j(x_i,\sigma^2)e^{-\frac{x_i}{\sigma^2}}\right)_{\substack{1\leq i \leq N \\ 1\leq j\leq N-2}} & \begin{matrix} e^{-\frac{x_1}{\lambda_1}} \\ \vdots \\ e^{-\frac{x_N}{\lambda_1}} \end{matrix}\end{array}
\right] d\lambda_1
\end{align}

Before going further, we need the following result, demonstrated in Appendix \ref{ap:A},
\begin{lemma}
	\label{le:1}
	Given a family $\{a_1,\ldots,a_N\}\in {\mathbb R}^N$, $N\geq 2$, and $b\in {\mathbb R}^\ast$, we have
	
\begin{equation}
		\det \left[ \begin{array}{c:c} 1 &  \\ \vdots & \left( \kappa_{j}(a_i,b)\right)_{\substack{1\leq i\leq N \\ 1\leq j \leq N-1}} \\ 1 &  \end{array}\right] = \frac1{b^{N(N-1)}}\prod_{i<j}(x_j-x_i)
	\end{equation}
\end{lemma}

By factorizing every row of the matrix by $e^{-\frac{x_i}{\sigma^2}}$ and developing the determinant on the last column, one obtains

\begin{align}
P_{{\bf Y}|I_1}({\bf Y}) &= \frac{e^{\sigma^2}\sigma^{2(N-1)(N-L-1)}}{N\pi^{LN}\prod_{i<j}(x_i-x_j)} \int_{\sigma^2}^{+\infty} \lambda_1^{N-L-1}e^{-\lambda_1} e^{-\frac{\sum_{i=1}^Nx_i}{\sigma^2}} (-1)^{N-1} \sum_{l=1}^N (-1)^{N+l} \frac{e^{-x_l\left( \frac{1}{\lambda_1} - \frac{1}{\sigma^2} \right)}}{\sigma^{2(N-1)(N-2)}} \prod_{\substack{i<j \\ i\neq l \\ j\neq l}}(x_i-x_j) \\
&= \frac{e^{\sigma^2-\frac{1}{\sigma^2}\sum_{i=1}^Nx_i}}{N\pi^{LN}\sigma^{2(N-1)(L-1)}} \sum_{l=1}^N (-1)^{l-1} \int_{\sigma^2}^{+\infty} \lambda_1^{N-L-1}e^{-\lambda_1} \frac{e^{-x_l\left( \frac{1}{\lambda_1} - \frac{1}{\sigma^2} \right)}}{\prod_{i<l}(x_i-x_l)\prod_{i>l}(x_l-x_i)}d\lambda_1 \\
&= \frac{e^{\sigma^2-\frac{1}{\sigma^2}\sum_{i=1}^Nx_i}}{N\pi^{LN}\sigma^{2(N-1)(L-1)}} \sum_{l=1}^N \frac{e^{\frac{x_l}{\sigma^2}}}{{\prod_{\substack{i=1 \\ i\neq l}}^N(x_l-x_i)}}\int_{\sigma^2}^{+\infty} \lambda_1^{N-L-1} e^{-\left(\lambda_1+\frac{x_l}{\lambda_1}\right)} d\lambda_1  \\
\end{align}
which finally gives

\begin{equation}
\label{eq:resSIMO}
P_{{\bf Y}|I_1}({\bf Y}) = \frac{e^{\sigma^2-\frac{1}{\sigma^2}\sum_{i=1}^Nx_i}}{N\pi^{LN}\sigma^{2(N-1)(L-1)}}\sum_{l=1}^N \frac{e^{\frac{x_l}{\sigma^2}}}{{\prod_{\substack{i=1 \\ i\neq l}}^N(x_l-x_i)}} J_{N-L-1}(\sigma^2,x_l) 
\end{equation}
with

\begin{align}
J_k(x,y) &= \int_{x}^{+\infty}t^{k}e^{-t-\frac{y}{t}}dt\\
&= 2y^{\frac{k+1}2}K_{-k-1}(2\sqrt{y}) - \int_{0}^{x}t^{k}e^{-t-\frac{y}{t}}dt
\end{align}
where $K_n$ denotes the modified Bessel function of the second kind.

We finally have the desired decision criterion

\begin{equation}
C_{{\bf Y}|I_1}({\bf Y}) = \frac1N\sum_{l=1}^N \frac{\sigma^{2(N+L-1)}e^{\sigma^2+\frac{x_l}{\sigma^2}}}{{\prod_{\substack{i=1 \\ i\neq l}}^N(x_l-x_i)}} J_{N-L-1}(\sigma^2,x_l) 
\end{equation}
\end{proof}

\subsubsection{Derivation of $P_{{\bf Y}|\mathcal H_i}$ in MIMO case}
In the MIMO configuration, $P_{{\bf Y}|\mathcal H_0}$ remains unchanged and equation \eqref{eq:PH0} is still correct. For the subsequent derivations, we only treat the situation when $M\leq N$ but the case $M>N$ is a trivial extension.

In this scenario, ${\bf H}\in \mathbb C^{N\times M}$ is, as already mentioned, distributed as a Gaussian i.i.d. matrix. The mean variance of every row is ${\rm E}[\sum_{j=1}^M|h_{ij}|^2]=1$. Therefore $M{\bf HH}^{\sf H}$ is distributed as a standard Wishart matrix. Hence, observing that ${\bf \Sigma}-\sigma^2 {\bf I}_N$ is the diagonal matrix of eigenvalues of ${\bf HH}^{\sf H}$,

\begin{equation}
  \label{eq:MIMOSigma}
  {\bf \Sigma}={\bf U}\cdot \diag(\nu_1+\sigma^2,\ldots,\nu_M+\sigma^2,\sigma^2,\ldots,\sigma^2) \cdot {\bf U}^{\sf H}
\end{equation}
the eigenvalue distribution density of $\bf \Lambda$ can be derived \cite{RMT}

\begin{align}
  \label{eq:MIMOLambda}
P_{\bf \Lambda}(\bf \Lambda) &= 
\frac{(N-M)!M^{MN}}{N!}
\prod_{i=1}^M e^{-M\sum_{i=1}^M(\lambda_i-\sigma^2)}\frac{(\lambda_i-\sigma^2)_+^{N-M}}{(M-i)!(N-i)!}\prod_{i<j}^M(\lambda_i-\lambda_j)^2\prod_{i>M}^N\delta(\lambda_i-\sigma^2)
\end{align}

From the equations \eqref{eq:MIMOSigma} and \eqref{eq:MIMOLambda} above, we will now show the MIMO equivalent result to theorem \ref{th:1}, which unfolds as follows,
\begin{theorem}
  \label{th:2}
The detection ratio $C_{ {\bf Y}|I_M}({\bf Y})$ for the presence of informative signal under prior information $I_M$, i.e. when the receiver is aware of (i) $M\leq N$ signal sources, (ii) the SNR $\sigma^{-2}$, reads

\begin{align}
\label{eq:C_MIMO}
C_{{\bf Y}|I_M}({\bf Y}) &=\frac{\sigma^{2M(N+L-M)}(N-M)!e^{M^2\sigma^2}}{N!M^{(M-1-2L)M/2}\prod_{j=1}^{M-1}j!} 
 \sum_{{\bf a}\subset [1,N]}\frac{e^{\frac{\sum_{i=1}^Mx_{a_i}}{\sigma^2}}}{\displaystyle\prod_{a_i}\displaystyle\prod_{\substack{j\neq a_1 \\ \ldots \\ j\neq a_i}}(x_{a_i}-x_j)}\sum_{{\bf b}\in \mathcal P(M)}(-1)^{\sgn({\bf b})+M}\prod_{l=1}^M J_{N-L-2+b_l}(M\sigma^2,Mx_{a_l})
\end{align}
with $\mathcal P(M)$ the ensemble of permutations of $\{1,\ldots,M\}$, ${\bf b}=(b_1,\ldots,b_M)$ and $\sgn({\bf b})$ the signature of the permutation $\bf b$.
\end{theorem}

\begin{proof}
  We note first that for any couple $(\lambda_i,\lambda_j)$ of the $M$ largest eigenvalues of $\bf \Sigma$, $j\neq i$, $P_{\bar{\bf \Lambda}}(\lambda_1,\ldots,\lambda_i,\ldots,\lambda_j,\ldots,\lambda_M)=P_{\bar{\bf \Lambda}}(\lambda_1,\ldots,\lambda_j,\ldots,\lambda_i,\ldots,\lambda_M)$, where $\bar{\bf \Lambda}$ is the joint (unordered) random variable $(\lambda_1,\ldots,\lambda_M)$. Since $\bf H$ is still isometric in the case $M>1$, those two conditions are sufficient \cite{GUI06} to ensure

\begin{align}
P_{{\bf Y}|I_M}({\bf Y}) 
&= \int_{{\bf \Sigma}} P_{{\bf Y}|{\bf \Sigma},I_M}({\bf Y},{\bf \Sigma})P_{\bf \Sigma}({\bf \Sigma})d{\bf \Sigma} \\
&= \label{eq:bigint2} \int_{\mathcal U(N)\times \mathbb R^{+^M}} P_{{\bf Y}|{\bf \Sigma},I_M}({\bf Y},{\bf \Sigma})P_{\bar{\bf \Lambda}}(\bar{\bf \Lambda})d{\bf U}d\bar{\bf \Lambda}
\end{align}
which, using the same technique as previously, is further developed into

\begin{align}
P_{{\bf Y}|I_M}({\bf Y}) &= \lim_{\lambda_{M+1},\ldots,\lambda_N\rightarrow \sigma^2} \frac{(N-M)!M^{MN}e^{M^2\sigma^2}\sigma^{2(N-M)(N-L-1)}\prod_{j=1}^{N-1-M}j!}{N!\pi^{NL}\prod_{j=1}^{M-1}j!} \nonumber \\
&\times \int_{\sigma^2}^{+\infty}\cdots\int_{\sigma^2}^{+\infty}\prod_{i=1}^M{\lambda_i}^{N-L-1}(\lambda_i-\sigma^2)^{N-M}\prod_{i<j}^M(\lambda_i-\lambda_j)^2e^{-M\sum_{i=1}^M\lambda_i}\frac{\det\left(\left\{e^{-\frac{x_i}{\lambda_j}}\right\}_{\substack{1\leq i\leq N \\ 1\leq j\leq N}}\right)}{\Delta({\bf X})\Delta({\bf \Lambda})} d\lambda_1\ldots d\lambda_M \\
&= \frac{(N-M)!M^{MN}e^{M^2\sigma^2}\sigma^{2(N-M)(N-L-1)}(-1)^{MN-\frac{M(M+1)}{2}}}{N!\pi^{NL}\prod_{j=1}^{M-1}j!} \nonumber \\
&\times \int_{\sigma^2}^{+\infty}\cdots\int_{\sigma^2}^{+\infty}\prod_{i=1}^M{\lambda_i}^{N-L-1}\frac{\prod_{i<j}^M(\lambda_i-\lambda_j)}{\prod_{i<j}^N(x_i-x_j)}e^{-M\sum_{i=1}^M\lambda_i}\det\left[
\begin{array}{c:c:c}
\begin{matrix} e^{-\frac{x_1}{\sigma^2}} \\ \vdots \\ e^{-\frac{x_N}{\sigma^2}} \end{matrix} & \kappa_j(x_i,\sigma^2)e^{-\frac{x_i}{\sigma^2}} & \begin{matrix} e^{-\frac{x_1}{\lambda_M}} & \cdots & e^{-\frac{x_1}{\lambda_1}} \\ \vdots & \cdots & \vdots \\ e^{-\frac{x_N}{\lambda_M}} & \cdots & e^{-\frac{x_N}{\lambda_1}} \end{matrix}
\end{array}
\right]
\end{align}
in which the term $(-1)^{MN-\frac{M(M+1)}{2}}$ originates from the $M$ exchanges between the $k^{th}$ column and the $(N-k+1)^{th}$ column, $k\in [1,M]$.

By factorizing the determinant by $e^{-\frac{1}{\sigma^2}\sum_{i=1}^Nx_i}$, developing along the $M$ last columns, we have from Lemma \ref{le:1},

\begin{align}
&\det\left[
\begin{array}{c:c:c}
	\begin{matrix} e^{-\frac{x_1}{\sigma^2}} \\ \vdots \\ e^{-\frac{x_N}{\sigma^2}} \end{matrix} & \kappa_j(x_i,\sigma^2)e^{-\frac{x_i}{\sigma^2}} & \begin{matrix} e^{-\frac{x_1}{\lambda_M}} & \cdots & e^{-\frac{x_1}{\lambda_1}} \\ \vdots & \cdots & \vdots \\ e^{-\frac{x_N}{\lambda_M}} & \cdots & e^{-\frac{x_N}{\lambda_1}} \end{matrix}
\end{array}
\right] \\
&= e^{-\frac{\sum_{i=1}^Nx_i}{\sigma^2}} \sum_{a_1\in [1,N]}(-1)^{\left[N+a_1\right]}e^{-x_{a_1}\left(\frac{1}{\lambda_1}-\frac{1}{\sigma^2}\right)}\cdots \nonumber \\
&~ ~ ~ ~ ~ ~ ~ ~ ~ ~ ~ \times \sum_{\substack{a_M\neq a_1 \\ \ldots \\ a_M\neq a_{M-1}}}^N(-1)^{\left[N-M+1+a_M-\sum_{i<M}\delta(a_i<a_M)\right]}\frac{e^{-x_{a_M}\left(\frac{1}{\lambda_M}-\frac{1}{\sigma^2}\right)}}{\sigma^{2(N-M-1)(N-M)}}\prod_{\substack{i<j \\ i,j\neq a_1 \\ \cdots \\ i,j\neq a_{M}}}(x_i-x_j) \\
&= e^{-\frac{\sum_{i=1}^Nx_i}{\sigma^2}} \sum_{{\bf a}\subset [1,N]}(-1)^{\left[ MN-\frac{M(M+1)}{2}+\sum_{i=1}^Ma_i+\sum_{i<j}\delta(a_i<a_j)\right]}e^{-\sum_{i=1}^Mx_{a_i}\left(\frac{1}{\lambda_i}-\frac{1}{\sigma^2}\right)}\prod_{\substack{i<j \\ i,j\neq a_1\\ \cdots \\ i,j\neq a_{M}}}(x_i-x_j) \\
&=  e^{-\frac{\sum_{i=1}^Nx_i}{\sigma^2}} \sum_{{\bf a}\subset [1,N]}(-1)^{\left[ MN-\frac{M(M+1)}{2}+\sum_{i=1}^Ma_i+\frac{M(M-1)}{2}\right]}\frac{e^{-\sum_{i=1}^Mx_{a_i}\left(\frac{1}{\lambda_i}-\frac{1}{\sigma^2}\right)}}{\sigma^{2(N-M-1)(N-M)}}\frac{(-1)^{(\sum_{i=1}^Ma_i)-M}\prod_{i<j}^N(x_i-x_j)}{\prod_{a_i}\prod_{\substack{j\notin [a_1,\ldots,a_i]}}(x_{a_i}-x_j)} \\
&= e^{-\frac{\sum_{i=1}^Nx_i}{\sigma^2}} \sum_{{\bf a}\subset [1,N]}\frac{e^{-\sum_{i=1}^Mx_{a_i}\left(\frac{1}{\lambda_i}-\frac{1}{\sigma^2}\right)}}{\sigma^{2(N-M-1)(N-M)}}\frac{\prod_{i<j}^N(x_i-x_j)}{\prod_{a_i}\prod_{\substack{j\notin [a_1,\ldots,a_i]}}(x_{a_i}-x_j)}
\end{align}

Together, this becomes,

\begin{align}
P_{{\bf Y}|I_M}({\bf Y}) &= \frac{M^{MN}(N-M)!e^{M^2\sigma^2-\frac{\sum_{i=1}^Nx_i}{\sigma^2}}}{N!\pi^{NL}\sigma^{2(N-M)(L-M)}\prod_{j=1}^{M-1}j!} \nonumber \\
&\times \int_{\sigma^2}^{+\infty}\cdots\int_{\sigma^2}^{+\infty}\prod_{i=1}^M{\lambda_i}^{N-L-1}\prod_{i<j}^M(\lambda_i-\lambda_j)e^{-M\sum_{i=1}^M\lambda_i}\sum_{{\bf a}\subset [1,N]}\frac{e^{-\sum_{i=1}^Mx_{a_i}\left(\frac{1}{\lambda_i}-\frac{1}{\sigma^2}\right)}}{\prod_{a_i}\prod_{\substack{j\notin [a_1,\ldots,a_i]}}(x_{a_i}-x_j)}d\lambda_1 \ldots d\lambda_M \\
&= \frac{(N-M)!e^{M^2\sigma^2-\frac{\sum_{i=1}^Nx_i}{\sigma^2}}}{N!M^{(M-2L-1)M/2}\pi^{NL}\sigma^{2(N-M)(L-M)}\prod_{j=1}^{M-1}j!} \nonumber \\
&\times \sum_{{\bf a}\subset [1,N]} \frac{e^{\frac{\sum_{i=1}^Mx_{a_i}}{\sigma^2}}}{\prod_{a_i}\prod_{\substack{j\notin \{a_1,\ldots, a_i\}}}(x_{a_i}-x_j)} \int_{M\sigma^2}^{+\infty}\cdots\int_{M\sigma^2}^{+\infty}e^{-\sum_{i=1}^M\left(\lambda_i+\frac{Mx_{a_i}}{\lambda_i}\right)}\prod_{i=1}^M{\lambda_i}^{N-L-1}\prod_{i<j}^M(\lambda_i-\lambda_j)d\lambda_1 \ldots d\lambda_M
\end{align}

Remind now the Vandermonde determinant identity

\begin{equation}
  \prod_{i<j}^M(X_j-X_i) = \sum_{ {\bf b}\in \mathcal P(M)}\sgn({\bf b})\prod_{i=1}^M X_i^{b_i-1}
\end{equation}
where $\mathcal P(k)$ is the ensemble of permutations of $k$ and $\sgn({\bf b})$ designs the signature of the permutation $\bf b$. Recognizing the expression of $J_k$, we finally obtain

\begin{align}
\label{eq:resMIMO}
P_{{\bf Y}|I_M}({\bf Y}) &=\frac{(N-M)!M^{(2L-M+1)M/2}e^{M^2\sigma^2-\frac{\sum_{i=1}^Nx_i}{\sigma^2}}}{N!\pi^{NL}\sigma^{2(N-M)(L-M)}\prod_{j=1}^{M-1}j!} \nonumber \\
&\times \sum_{{\bf a}\subset [1,N]}\frac{e^{\frac{\sum_{i=1}^Mx_{a_i}}{\sigma^2}}}{\displaystyle\prod_{a_i}\displaystyle\prod_{\substack{j\neq a_1 \\ \ldots \\ j\neq a_i}}(x_{a_i}-x_j)}\sum_{{\bf b}\in \mathcal P(M)}(-1)^{\sgn({\bf b})+M}\prod_{l=1}^M J_{N-L-2+b_l}(M\sigma^2,Mx_{a_i})
\end{align}

For instance, when $M=2$, we have

\begin{align}
P_{{\bf Y}|I_2}({\bf Y}) &= \frac{2^{2L-1}e^{4\sigma^2-\frac{1}{\sigma^2}\sum_{i=1}^Nx_i}}{N(N-1)\sigma^{2(N-2)(L-2)}\pi^{NL}} \sum_{{\bf a}\subset [1,N]}\frac{e^{\frac{x_{a_1}+x_{a_2}}{\sigma^2}}}{\prod_{\substack{j\neq a_1}}(x_{a_1}-x_j) \prod_{\substack{j\neq a_i\\ j\neq a_2}}(x_{a_2}-x_j)}\left(J_{N-L}^{(a_1)}J_{N-L-1}^{(a_2)}-J_{N-L-1}^{(a_1)}J_{N-L}^{(a_2)}\right)
\end{align}
in which $J^{(x)}_{k}=J_{k}(2\sigma^2,2x)$.

Decisions regarding the signal detection are then carried out by computing the ratio $C_{ {\bf Y} | I_M}({\bf Y})$ between equation \eqref{eq:resMIMO} and equation \eqref{eq:PH0} as follows

\begin{align}
%\label{eq:C_MIMO}
C_{{\bf Y}|I_M}({\bf Y}) &=\frac{\sigma^{2M(N+L-M)}(N-M)!e^{M^2\sigma^2}}{N!M^{(M-1-2L)M/2}\prod_{j=1}^{M-1}j!} 
 \sum_{{\bf a}\subset [1,N]}\frac{e^{\frac{\sum_{i=1}^Mx_{a_i}}{\sigma^2}}}{\displaystyle\prod_{a_i}\displaystyle\prod_{\substack{j\neq a_1 \\ \ldots \\ j\neq a_i}}(x_{a_i}-x_j)}\sum_{{\bf b}\in \mathcal P(M)}(-1)^{\sgn({\bf b})+M}\prod_{l=1}^M J_{N-L-2+b_l}(M\sigma^2,Mx_{a_i})
\end{align}
\end{proof}

Note that $C_{{\bf Y}|I_M}({\bf Y})$ is a function of the empirical eigenvalues $x_1,\ldots,x_N$ of ${\bf YY}^{\sf H}$ only. This is explained by the presence of only Gaussian random entities, whose isometric property leads the eigenvectors of ${\bf YY}^{\sf H}$ to contain no additional information. Remark also that equation \eqref{eq:C_MIMO} is a non-trivial function of $x_1,\ldots,x_N$, which does not involve the sum of the $x_i$'s as for the case of the classical energy detector.

In the following, we extend the current signal detector to the situations where $M$ and $\sigma^2$ are not {\it a priori} known at the receiver.

\subsection{Number of sources and/or noise variance unknown}
\subsubsection{Unknown noise variance}
\label{sec:SNR}
Efficient signal detection when the noise level is unknown is highly desirable. Indeed, if the noise level were exactly known, some prior noise detection mechanism would be required. The difficulty here is handily avoided thanks to {\it ad-hoc} methods that are asymptotically independent of the noise level \cite{CAR08}-\cite{ZEN08}.
Instead, we shall consider some prior information about the noise level. Establishing prior information of variables defined in a continuum is still a controverted debate of the maximum entropy theory. However, a few solutions are classically considered that are based on desirable properties. Those are successively detailed in the following.

Two classical cases are usually encountered,
\begin{itemize}
\item the noise level is known to belong to a continuum $[\sigma^2_{-},~\sigma^2_{+}]$. If no more information is known, then it is desirable to take a uniform prior for $\sigma^2$ and then

\begin{equation}
\label{eq:unis2}
P_{\sigma^2}(\sigma^2)d{\sigma^2}=\frac{1}{\sigma^2_{+}-\sigma^2_{-}}d{\sigma^2}
\end{equation}
However, a questionable issue of invariance to variable change arises. Indeed, if $P_{\sigma^2}$ is uniform, the distribution associated to the variable $\sigma=\sqrt{\sigma^2}$ is then non-uniform. This old problem is partially answered by Jeffreys \cite{JEY46} who suggests that an {\it uninformative prior} should be any distribution that does not add information to the posterior distribution $P_{\sigma^2|{\bf Y},I_M}$ (for recent developments, see also \cite{CAT01}).
However, in our problem, the uninformative prior is rather involved so we only consider uniform prior distribution \eqref{eq:unis2} for $\sigma^2$ (we denote $I_M'=``\mathcal H_1,\sigma^2\in [\sigma^2_{-},~\sigma^2_{+}]"$) and therefore

\begin{align}
\label{eq:intsigma}
P_{{\bf Y}|I_M'} &= \frac{1}{\sigma^2_{+}-\sigma^2_{-}}\int_{\sigma^2_{-}}^{\sigma^2_{+}} P_{{\bf Y}|\sigma^2,I_M'}({\bf Y},\sigma^2)d\sigma^2 
\end{align}
\item one has no information concerning the noise power. The only information about $\sigma^2$ is $\sigma^2>0$. Again, we might want to subjugate $\sigma^2$ to Jeffreys' {\it uninformative prior}. However, computing this prior is again rather involved. The other alternative is to take the limit of \eqref{eq:intsigma} when $\sigma_{-}$ tends to zero and $\sigma_{+}$ tends to infinity. This limiting process produces an improper integral form. This would be, with $I_M''$ the updated background information,

\begin{align}
\label{eq:intsigma2}
P_{{\bf Y}|I_M''} &= \lim_{x\rightarrow \infty}\frac{1}{x-\frac{1}{x}}\int_{\frac{1}{x}}^{x} P_{{\bf Y}|\sigma^2}({\bf Y},\sigma^2)d\sigma^2
\end{align}
\end{itemize}

This leads in any case to the updated decisions $C_{ {\bf Y}|M}$ of the form,

\begin{equation}
  C_{ {\bf Y}|M} = \frac{\int_{\sigma^2_-}^{\sigma^2_+}P_{ {\bf Y}|\sigma^2,I_M'}({\bf Y},\sigma^2)d\sigma^2}{\int_{\sigma^2_-}^{\sigma^2_+}P_{ {\bf Y}|\sigma^2,\mathcal H_0}({\bf Y},\sigma^2)d\sigma^2}
\end{equation}

The computational difficulty raised by the integrals $J_k(x,y)$ does not allow for any satisfying closed-form formulas for \eqref{eq:intsigma} and \eqref{eq:intsigma2}. 
In the following, we therefore only consider the bounded continuum scenario.

\begin{figure}
\centering
\includegraphics[]{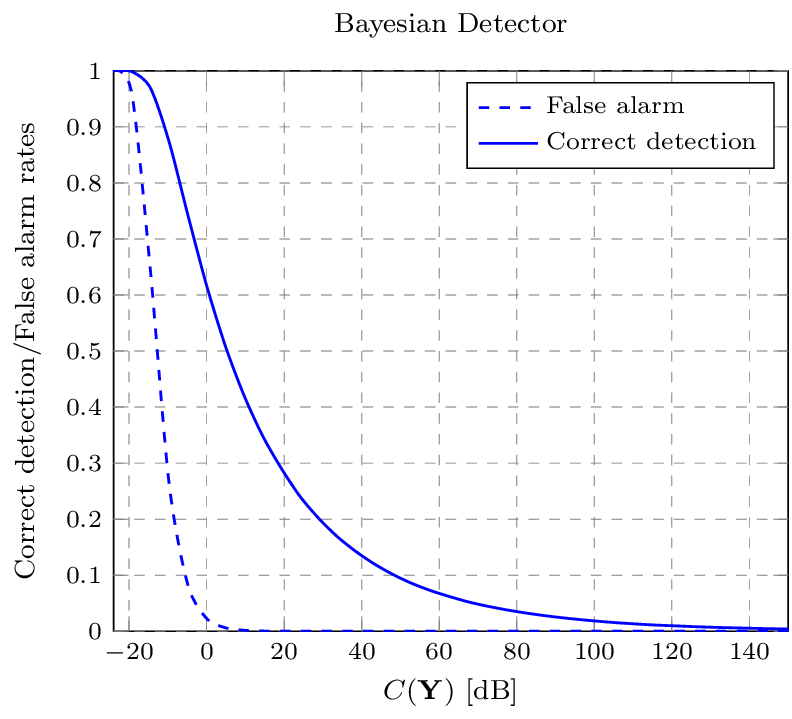}
\includegraphics[]{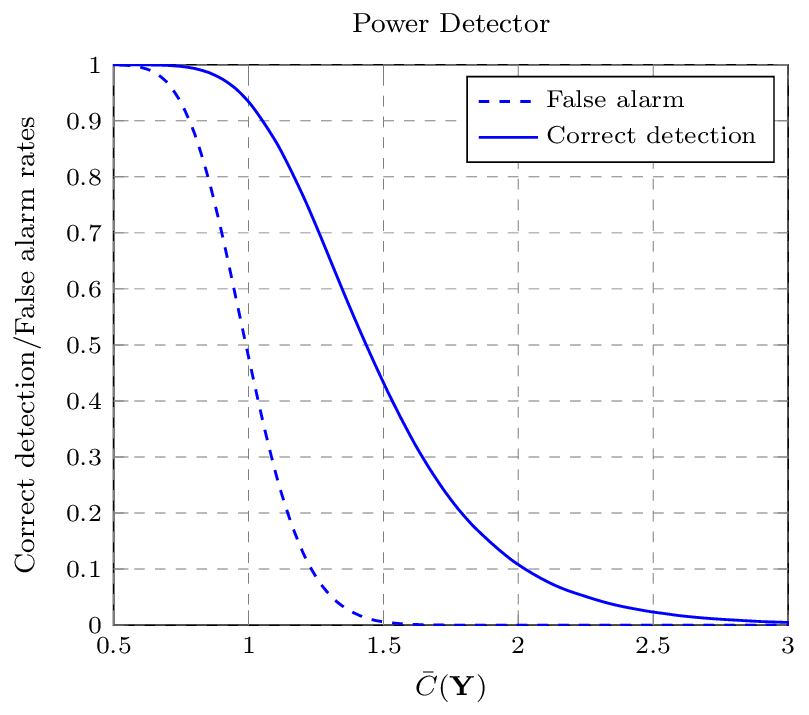}
\caption{Power detection performance in SIMO - $M=1$, $N=4$, $L=8$, ${\sf SNR}=-3~{\rm dB}$. On the left, Bayesian detector; on the right, classical power detector.}
\label{fig:FA1}
\end{figure}

\subsection{Unknown number of sources $M$}
In practical cases, the number of transmitting sources is only known to be finite. If only an upper bound value $M_{\rm max}$ on $M$ is known, a uniform prior is assigned to $M$ (which is again compliant with the maximum entropy principle). The probability distribution of $\bf Y$ under hypothesis $I_0=$``$\sigma^2$ known, $M$ unknown'', reads

\begin{align}
P({\bf Y}|I_0) &= \sum_{i=1}^{M_{\rm max}}P({\bf Y}|``M=i'',I_0)\cdot P(``M=i''|I_0) \\
&= \frac{1}{M_{\rm max}}\sum_{i=1}^{M_{\rm max}}P({\bf Y}|``M=i'',I_0)
\end{align}
which does not meet any computational difficulty. 

This leads then to the decision ratio $C_{ {\bf Y}|\sigma^2}$,

\begin{equation}
  C_{ {\bf Y}|\sigma^2} = \frac{\sum_{i=1}^{M_{\rm max}} P({\bf Y}|``M=i'',I_0) }{\sum_{i=1}^{M_{\rm max}} P({\bf Y}|``M=i'',\mathcal H_0)}
\end{equation}

\section{Simulation and Results}
\label{sec:simu}

\begin{figure}
\centering
\includegraphics[]{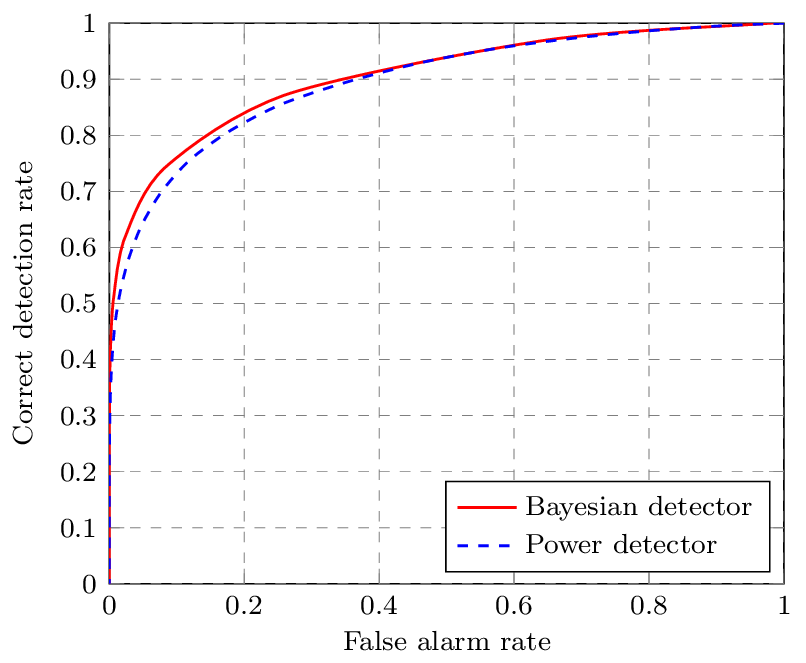}
\includegraphics[]{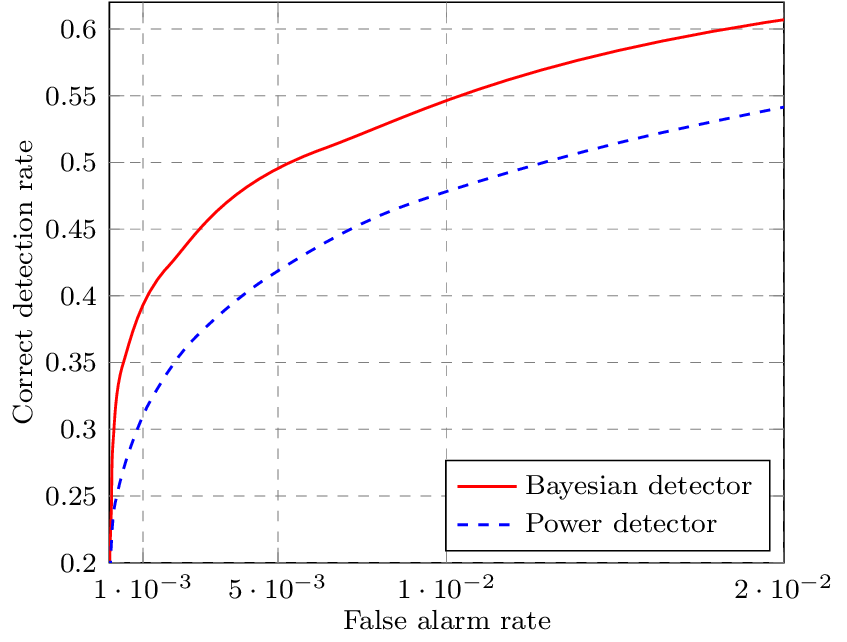}
\caption{CDR against FAR for SIMO transmission - $M=1$, $N=4$, $L=8$, ${\sf SNR}=-3~{\rm dB}$. On the left, full FAR range; on the right, FAR range of practical interest.}
\label{fig:CDmFA1}
\end{figure}

\begin{figure}
\centering
\includegraphics[]{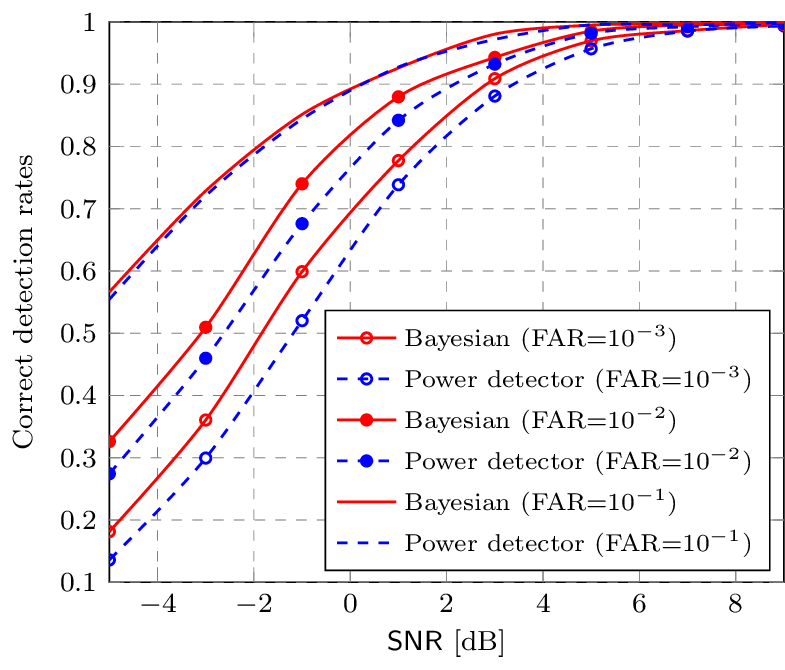}
\caption{Correct detection rates under FAR constraints for different SNR levels, $M=1$, $N=4$, $L=8$}
\label{fig:CD_SNR}
\end{figure}

\begin{figure}
  \centering
\includegraphics[]{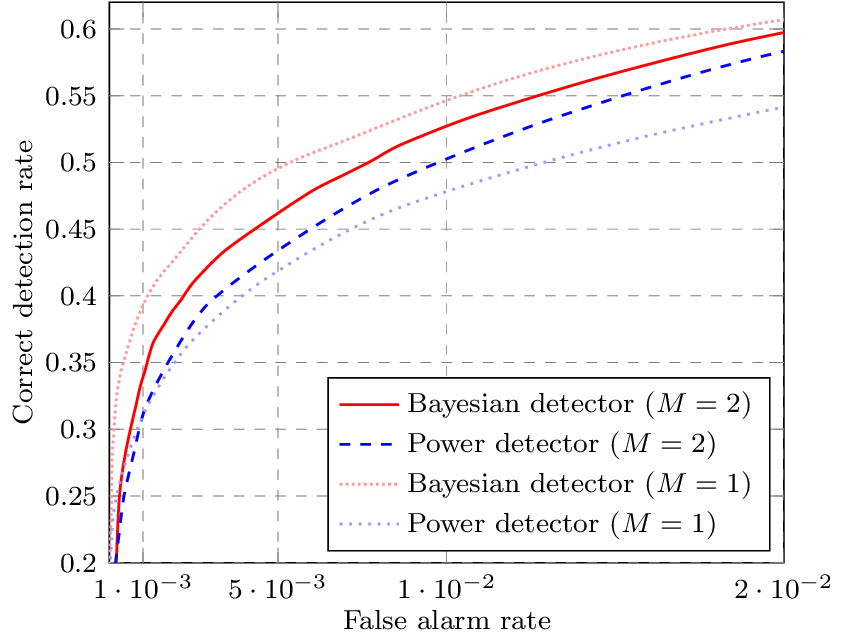}
\caption{CDR against FAR for SIMO transmission - $M=2$, $N=4$, $L=8$, ${\sf SNR}=-3~{\rm dB}$. FAR range of practical interest.}
\label{fig:CDmFA2}
\end{figure}

In the following, we present results obtained for the aforementioned SIMO and MIMO scenarios, using theorems \ref{th:1} and \ref{th:2} respectively.
In the simulations, the hypothesis concerning incoming data, channel aspect and noise figure are those presented in the model of Section \ref{sec:model},\footnote{it might be objected that assuming the maximum entropy distributions in the simulations is dishonest since the real distributions are unknown; the simulations here however intend only to verify the validity of our theoretical results and cannot be used as any proof of performance.} i.e. the channel, signal and noise matrix entries are i.i.d. Gaussian with respective variance $1/M$, $1$ and $1$. The results are validated by and compared with the classical power detector, which merely consists in summing the eigenvalues $x_1,\ldots,x_N$ of ${\bf YY}^{\sf H}$ and assumes an empirical decision threshold; the corresponding decision is then based on the scaled value

\begin{equation}
  \bar{C}({\bf Y}) = \frac1{LN\sigma^2}\sum_{i=1}^Nx_i
\end{equation}

In our first example, we consider a SIMO channel with $N=4$ antennas at the receiver, $L=8$ sampling periods and a signal to noise ratio ${\sf SNR}=-3~{\rm dB}$. For fair comparison with classical signal detection algorithms, we stick to the \textit{false alarm rate} (FAR) against \textit{correct detection rate} (CDR) performance evaluation. Figure \ref{fig:FA1} presents the respective FAR and CDR for the novel Bayesian estimator and for the classical power detector, obtained on $100,000$ Monte Carlo channel realizations. Depending on the acceptable FAR, the decision threshold for the power detector is somewhere around $1$ while the threshold for the Bayesian approach is somewhere around $C({\bf Y})=0~{\rm dB}$. Since both algorithms scale very differently, fair comparison is obtained by plotting the CDR versus FAR curve. This is depicted in Figure \ref{fig:CDmFA1}. For large values of the FAR, both Bayesian and power detectors produce similar performance. In most practical applications, one is however interested into high performance CDR for fixed low FAR (in the order of $10^{-3}$ to $10^{-1}$ depending on the application); for this range of FAR, it is observed, in the right part of Figure \ref{fig:CDmFA1}, that as much as a $10\%$ increase in detection ability is obtained by the Bayesian detector and this gain increases along with smaller FAR. This is confirmed by Figure \ref{fig:CD_SNR} in which the performance of the Bayesian signal detector with respect to the energy detector for different constrained FAR is presented against the SNR.

In Figure \ref{fig:CDmFA2}, we took $N=4$, $L=8$ and ${\sf SNR}=-3~{\rm dB}$ as before but consider now $M=2$ signal sources; we then use theorem \ref{th:2} here. In this scenario the classical power detector closes in the gap with the Bayesian detector, compared to the SIMO situation. Having depicted in the same plot the detection performance for $M=1$, we notice that the classical energy detector performs better in the scenario $M=2$, which might be interpreted as a result of a channel hardening effect \cite{HOC04}; on the contrary, the Bayesian detector performs less accurately in this case, which can be attributed to the increased number of random variables (both in the channel matrix $\bf H$ and in the signal matrix $\bf \Theta$) inducing increased `freeness' in the interpretation of the hypothetical origins of the output matrix $\bf Y$.

Consider now the scenario when the noise variance $\sigma^2$ is {\it a priori} known to belong to the interval $[\sigma^2_{-},\sigma^2_{+}]$. The two-dimensional integration of equation \eqref{eq:unis2} is prohibitive for producing numerical results. We therefore divide the continuum $[\sigma^2_{-},\sigma^2_{+}]$ into $K$ subsets $[\sigma^2_{-}+k\Delta(\sigma^2),\sigma^2_{-}+(k+1)\Delta(\sigma^2)]$, for $k\in \{0,\ldots,K-1\}$ and $\Delta(\sigma^2)=(\sigma^2_{+}-\sigma^2_{-})/K$ where $\Delta(\sigma^2)$ is chosen small enough so to produce a rather good approximation of \eqref{eq:unis2}. This is presented in Figure \ref{fig:CDCDSNR} which demonstrates the effect of an inaccurate knowledge of the noise power in terms of CDR and FAR. In this simulation, $M=1$, $N=4$, $L=8$ and ${\sf SNR}=0~{\rm dB}$. Comparison is made between the cases of exact SNR knowledge, short SNR range $[\sigma^2_{-},\sigma^2_{+}]=[-2.5,2.5]$~dB discretized as a set $\{-2.5,-1.5,\ldots,1.5,2.5\}$~dB, large SNR range $[\sigma^2_{-},\sigma^2_{+}]=[-5,5]$~dB discretized as a set $\{-5,-4,\ldots,4,5\}$~dB and very large range $[\sigma^2_{-},\sigma^2_{+}]=[-9,9]$~dB discretized as a set $\{-9,-8,\ldots,8,9\}$~dB. Observe that the short SNR range provides already a strong performance decay compared to the ideal scenario, which is particularly noticeable in terms of CDR performance at low FAR. Larger SNR ranges are then only slightly worse than the short range scenario and seem to converge to a `worst-case limit';\footnote{note that the plot for the very large SNR range was intentionally removed from the left hand-side plot for it almost perfectly fitted to plot corresponding to the large SNR range.} this can be interpreted by the fact that the additional hypotheses, i.e. very strong or very little noise power, are automatically discarded as the values of $P_{ {\bf Y}|\sigma^2,I_1'}({\bf Y},\sigma^2)$ and $P_{ {\bf Y}|\sigma^2,\mathcal H_0}({\bf Y},\sigma^2)$ become negligible for unrealistic values of $\sigma^2$. Additional simulations for larger SNR ranges were carried out that visually confirm that the FAR and CDR plots are identical here as long as $\sigma^2_-\leq -5$ dB and $\sigma^2_+\geq 5$ dB. Therefore, simulations suggest that the proposed Bayesian signal detector is able to cope even with totally unknown SNR, which is obviously not the case of the classical energy detector that relies on an SNR-dependent decision threshold.

\begin{figure}
\centering
\includegraphics[]{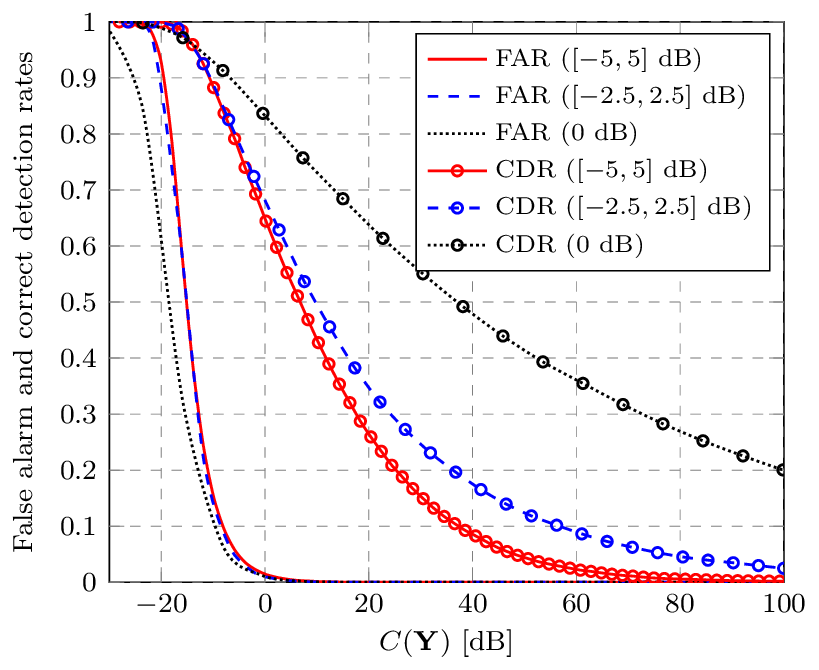}
\includegraphics[]{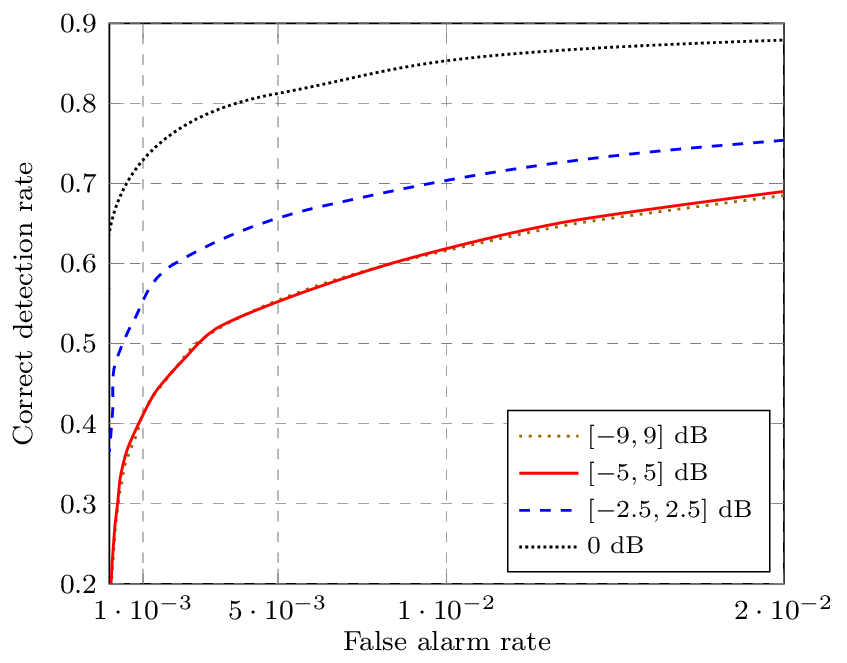}
\caption{FAR and CDR, for different {\it a priori} information: exact SNR ($0$ dB), short range SNR ($[-2.5,2.5]$ dB) and large range SNR ($[-5,5]$ dB) SNR, $M=1$, $N=4$, $L=8$, true ${\sf SNR}=0~{\rm dB}$}
\label{fig:CDCDSNR}
\end{figure}

\section{Discussion}
\label{sec:discussion}
In the previous framework, we relied on the maximum entropy principle in order to derive unique {\it a priori} distributions for the various unknown system parameters. The provided Bayesian solutions, derived from the channel state of knowledge available at the receiver, were claimed consistent in the proposed probability framework. This framework is in particular extensible to whatever prior knowledge the receiver might have on the transmission environment. However, some limitations can be raised. First, as stated in \ref{sec:SNR}, uninformative priors modeling is still an incomplete and controversial theory, for which no definite answer is available to this day. When such a prior information is to be treated, the proposed signal detection framework is not capable of singling out a proper maximum entropic model; this constitutes a major coherency issue of our source detection framework. Also, the mathematical tools to derive maximum entropy distributions, e.g. Lagrangian multipliers, only cope with statistical prior knowledge, such as the moments of the underlying density functions, and are rarely able to treat deterministic knowledge.

Note however that the advances in the field of random matrix theory provides new answers to problems of high dimensionality, even for finite $N,L<\infty$ values. Those problems, such as the present maximum-likelihood multi-antenna signal detection, are often considered intractable and suffer in practice from the so-called curse of dimensionality. The current study relies nonetheless on the important property that the transmission channel $\bf H$ is modelled as i.i.d. Gaussian; if $\bf H$ were more structured, it would have been more difficult to obtain an integral expression similar to \eqref{eq:bigint2} and the final results, be it derivable, would implicate not only the eigenvalues but also the eigenvectors of ${\bf YY}^{\sf H}$. 

More importantly, the proposed Bayesian framework allows one to answer a wider scope of problems than the present multi-source signal detection. In particular, we discussed in the introduction the somewhat different problem of counting the number of transmitting sources, which has received a lot of interest. In our framework, this consists in considering a set of hypotheses $\mathcal H'_0,\ldots,\mathcal H'_{M_{\rm max}}$, where $\mathcal H'_k$ is the hypothesis ``$k$ signals are being transmitted'' for a maximum of $M_{\rm max}$ sources, and evaluating the $M_{\rm max}$ hypothesis tests

\begin{equation}
	C_{k|{\bf Y}}({\bf Y}) = \frac{P_{\mathcal H'_k|{\bf Y}}({\bf Y})}{\sum_{i\neq k}^{M_{\rm max}}P_{\mathcal H'_i|{\bf Y}}({\bf Y})} = \frac{P_{{\bf Y}|``M=k''}({\bf Y})P_{``M=k''}}{\sum_{i\neq k}^{M_{\rm max}}P_{{\bf Y}|``M=i''}({\bf Y})P_{``M=i''}} 
\end{equation}
which can be evaluated using the results derived above.

\section{Conclusion}
In this work, we introduced a general Bayesian framework for multi-antenna multi-source detection, which can be used to detect a single multi-antenna source emitter as well as multiple single antenna transmitters. This framework is based on a consistent treatment of the system information available at the sensing device. The performance of the novel Bayesian multi-signal source detector is derived and compared in simulations with the classical power detection technique. We observed in particular that the proposed Bayesian detector performs better than the classical energy detector for low false alarm rate constraints, and is capable of treating the problem of signal detection under imprecise or completely unknown signal-to-noise ratio.

\appendices
\section{Proof of Lemma \ref{le:1}}
\label{ap:A}
We will show the result by recursion. For $N=2$, the determinant is simply $(a_2-a_1)/b^2$ which is compliant with the formula. Assuming the result for $N-1$, we develop the determinant for dimension $N$ on the last column (that corresponding to $j=N-1$), to obtain

\begin{align}
	\sum_{i=1}^N(-1)^{N+i}\sum_{m=1}^{N-1}\frac{(-1)^{N-1+m}(N-1)!(N-2)!}{m!(N-1-m)!(m-1)!b^{N-1+m}}x_i^m\frac1{b^{(N-2)(N-1)}}\prod_{\substack{a<b \\ a\neq i\\ b\neq i}}(x_b-x_a)
\end{align}
which, after development of the Vandermonde determinant leads to

\begin{align}
	\sum_{i=1}^N(-1)^{N+i}\sum_{m=1}^{N-1}\frac{(-1)^{N-1+m}(N-1)!(N-2)!}{m!(N-1-m)!(m-1)!b^{N-1+m}}x_i^m\sum_{\sigma_i\in \mathcal P_i(N-1)}\sgn(\sigma_i)b^{-(N-2)(N-1)}\prod_{\substack{k=1\\ k\neq i}}^N x_k^{\sigma_i(k)-1}
\end{align}
where $\mathcal P_i(N-1)$ is the set of all functions that map $\{1,\ldots,i-1,i+1,\ldots,N\}$ onto $\{1,\ldots,N-1\}$ (this can be thought of as the set of permutations of $N-1$ with $i$ in input labelled by $N$).
Now notice that in the expression above, for $1\leq m\leq N-2$, one can exchange $x_i^m$ and $x_k^{\sigma_i(k)-1}$ for the unique $k$ such that $\sigma_i(k)-1=m$. However, this comes along with a change of the term $(-1)^{N+i}$ by $(-1)^{N+k}$ (hence a multiplication by $(-1)^{k-i}$) and a change of the signature of the `pseudo-permutation' $\sigma_i$ as $\sgn(\sigma_i) (-1)^{k-i-1}$ ($k-i-1$ being the number of steps needed to bring $k$ before $i+1$ by successive neighbor-permutations); the rest of the expression is not affected. Therefore, those two instances sum up to $0$ in the above expressions, for any $m\leq N-2$. This leads to consider only the term $N-1$ in the sum $\sum_{m=1}^{N-1}$. This reduces the above expression to

\begin{align}
	\sum_{i=1}^N(-1)^{N+i}b^{-N(N-1)} x_i^{N-1}\sum_{\sigma_i\in {\mathcal P}_i(N-1)}\sgn(\sigma_i)\prod_{\substack{k=1 \\ k\neq i}}^N x_k^{\sigma_i(k)-1} =
	\label{eq:lesum}	\sum_{i=1}^Nb^{-N(N-1)}\sum_{\sigma_i\in \bar{\mathcal P}_i(N)}\sgn(\sigma_i)\prod_{k=1}^N x_k^{\sigma_i(k)-1}
\end{align}
with $\bar{\mathcal P}_i(N)$ the set of permutations $\sigma$ of $N$ for which $\sigma(i)=N$ (which imposes a product by $(-1)^{N-i}$ to the previous signature and the collapse of the remaining $(-1)^{N+i}$ term). Now, ${\mathcal P}(N)$, the set of all permutations of $N$ satisfies

\begin{align}
	{\mathcal P}(N) &= \bar{\mathcal P}_i(N) \cup \left\{\sigma\in {\mathcal P}(N)| \sigma(i)\neq N \right\} \\
	&= \bar{\mathcal P}_i(N) \cup \bigcup_{k=1}^{N-1}\left\{\sigma\in {\mathcal P}(N)| \sigma(i)=k \right\}
\end{align}
the right-hand side member of the above equation is obtained by all the elements in the sum of \eqref{eq:lesum}. Hence the final result.

\end{document}